\documentclass[11pt]{article} 
\usepackage{fullpage}

\usepackage{bm}
\usepackage{amsthm}

\usepackage{thmtools}
\usepackage{thm-restate}

\usepackage{cleveref}

\usepackage{amsmath,amssymb}
\usepackage{fancybox}
\def\MM{{\mathbf{M}}}
\def\AA{\mathbf{A}}

\def\CC{\mathbf{C}}
\def\DD{{\mathbf{D}}}

\newdimen\pIR
\newcommand\R{{\rm I\kern\pIR R}}

\def\Reals#1{\mathbb{R}^{#1}}

\newenvironment{fminipage}%
 {\begin{Sbox}\begin{minipage}}%
 {\end{minipage}\end{Sbox}\fbox{\TheSbox}}
\newenvironment{algbox}[0]{\vskip 0.2in
\noindent 
\begin{fminipage}{6.3in}
}{
\end{fminipage}
\vskip 0.2in
} 
\newcommand{\Half}{\frac{1}{2}}

\newtheorem{theorem}{Theorem}[section]
\newtheorem{definition}[theorem]{Definition}
\newtheorem{lemma}[theorem]{Lemma}
\newtheorem{fact}[theorem]{Fact}
\newtheorem{proposition}[theorem]{Proposition}

\newcommand{\M}[1]{\mathbf{#1}}

\newcommand{\V}[1]{\bm{#1}}
\newcommand{\T}{\top}

\newcommand{\OO}[1]{O(#1)}

\newcommand{\Mij}[3]{#1_{{#2},{#3}}}

\newcommand{\numberthis}{\addtocounter{equation}{1}\tag{\theequation}}

\usepackage{color}
\newcommand{\algname}{\textsc{GraphSampling}}

\title{Spectral Sparsification of Random-Walk Matrix Polynomials}

\author{Dehua Cheng
  \footnotemark[1]
  \\USC \and
  Yu Cheng
  \footnotemark[2]
  \\USC \and
  Yan Liu
  \thanks{Supported in part by NSF research grants IIS-1134990, IIS-1254206 and U.S. Defense Advanced Research Projects Agency (DARPA) under Social Media in Strategic Communication (SMISC) program, Agreement Number W911NF-12-1-0034.}
  \\USC \and
  Richard Peng
  \\MIT \and
  Shang-Hua Teng
  \thanks{Supported in part by NSF grants CCF-1111270 and CCF-096448 and by the Simons Investigator Award from the Simons Foundation.}
  \\ USC}

\begin{document}

\maketitle

\begin{abstract}

We consider a fundamental algorithmic question in 
  spectral graph theory:
Compute a {\em spectral sparsifier} of
  a {\em random-walk matrix-polynomial}
$$\M{L}_{\V{\alpha}}(G) = \M{D} - \sum_{r=1}^d \alpha_r\M{D} \cdot 
  \left(\M{D}^{-1}\M{A}\right)^r$$
where $\M{A}$ is the adjacency matrix of a weighted, 
  undirected graph, $\M{D}$ is the diagonal matrix of 
  weighted degrees, and  
  $\V{\alpha} = (\alpha_1,...,\alpha_d)$ are nonnegative coefficients
  with $\sum_{r=1}^d \alpha_r = 1$.
Recall that $\M{D}^{-1}\M{A}$ is the transition matrix of
  random walks on the graph.
In its linear form (when $d = 1$), 
 the matrix polynomial becomes $\M{D} - \M{A}$, which is a {\em Laplacian matrix},
 and hence this problem becomes the standard
  spectral sparsification problem, which enjoys nearly linear time
  solutions \cite{SpielmanTengSpectralSparsification,SpielmanSrivastava}.
However, the sparsification of $\M{L}_{\V{\alpha}}(G)$
  appears to be algorithmically challenging as the matrix power 
  $(\M{D}^{-1}\M{A})^r$
  is defined by all paths of length $r$,
  whose precise calculation would be prohibitively expensive 
  (due to the cost of matrix multiplication and densification in the matrix powers).

In this paper, we develop the first nearly linear time algorithm 
  for this sparsification problem:
For any $G$ with $n$ vertices and $m$ edges,
  $d$ coefficients $\V{\alpha}$,
  and $\epsilon > 0$,
  our algorithm runs in time 
$O(d^2\cdot m\cdot \log^2 n/\epsilon^{2})$
  to construct a Laplacian  matrix $\tilde{\M{L}} =\M{D}-\tilde{\M{A}}$ with 
   $O(n\log n/\epsilon^{2})$
   non-zeros such that
$$\tilde{\M{L}} \approx_{\epsilon} \M{L}_{\V{\alpha}}(G) =  \M{D} - \sum_{r=1}^d \alpha_r\M{D} \cdot  \left(\M{D}^{-1}\M{A}\right)^r.$$
In the equation, $\tilde{\M{L}} \approx_{\epsilon} \M{L}_{\V{\alpha}}(G)$ 
denotes that $\tilde{\M{L}}$ and $\M{L}_{\V{\alpha}}(G)$
  are {\em spectrally similar} within a factor of $1\pm \epsilon$ 
  as defined in \cite{SpielmanTengSpectralSparsification}.

Matrix polynomials arise in 
  mathematical analysis of matrix functions as well as
  numerical solutions (such as Newton's method)
  of matrix equations.
Our work is particularly motivated by the algorithmic problems
  for speeding up the classic Newton's method
  in applications such as computing the inverse 
  square-root of the precision matrix
  of a Gaussian random field (in order to obtain {\em i.i.d}
  random samples of the graphic model), as well as 
  computing the $q\textsuperscript{th}$-root transition (for $q\geq 1$) in a time-reversible 
  Markov model. 
The key algorithmic step for both applications
  is the construction of a spectral sparsifier
  of a constant degree\footnote{In numerical algorithms
  where random-walk matrix-polynomials arise, the degree $d$ of the polynomials
  is usually either a constant or bounded above by a polylogarithmic
  function in $n$.} random-walk matrix-polynomials 
  introduced by Newton's  method.
Our sparsification algorithm leads to a simpler and faster algorithm
  for these problems than the previous one \cite{ChengCLPT14} that 
  circumvents the challenging problem of sparsifying high-degree random-walk 
  matrix polynomials at the cost of slower convergences and complex approximation.
Our algorithm can also be used to 
  build efficient data structures
  for effective resistances for multi-step time-reversible Markov models,
  and we anticipate that it could be useful for other 
  tasks in network analysis.

\end{abstract}


\section{Introduction}

Polynomials are used in many fields of mathematics and science
  for encoding equations that model various physical,
  biological and economical processes.
In scientific computing and its underpinning numerical analysis,
  polynomials appear naturally in (truncated) Taylor series,
  the fast multipole method \cite{Greengard}, and various numerical approximations.
These computational methods are responsible for a large part 
  of engineering and scientific simulations ranging from weather forecasting
  to earthquake modeling \cite{StrangFix} to particle/galaxy simulation \cite{Greengard}.

Like its scalar counterpart, 
  matrix polynomials of the form, $\sum_{i=0}^d c_i \cdot\M{M}^i$,
  arise in mathematical analysis of matrix functions and dynamical systems,
  as well as numerical solution of matrix equations.
One class of matrices of particular importance in network analysis is the 
 adjacency matrix of a weighted, undirected graph $G$.
We will denote these matrices using $\M{A}$.
If we use $\M{D}$ to denote the diagonal matrix containing 
  weighted degrees of vertices,
  $\M{D}^{-1}\M{A}$ is the transition matrix of
  random walks on the graph.
Powers of this matrix correspond to multiple steps of random walks on the
  graph, and are graphs themselves.
However, they are usually dense, and are cost-prohibitive
  both in time and memory to construct when the input graph
  $G$ is of large-scale.
Our objective is to efficiently construct sparse approximations 
  of this natural family of  random-walk matrices.

\subsection{Motivation from Gaussian Sampling and Fractional Markov Transition}
A problem that motivates our study of random walk polynomials is
the inverse square-root problem: find a linear operator $\CC$ such that $\CC^\T \CC$
is close to the inverse of the {\em Laplacian matrix}.
Laplacian matrices are a subclass of {\em SDDM matrices}
  \footnote{SDDM matrices are positive definite symmetric diagonal dominant matrices with non-positive off-diagonal elements.}
  and have standard split-form representation of $\M{L} = \M{D} - \M{A}$.
When we apply an extension of the classical Newton's method to
  this form we can reduce the problem
  to that of factoring 
$\M{D} - \left(\frac{3}{4} \M{D}\cdot(\M{D}^{-1} \M{A})^2 
 + \frac{1}{4}\M{D} \cdot(\M{D}^{-1} \M{A})^3\right)$,
  which has smaller spectral radius, by using 
  the matrix identity
\begin{equation}\label{iter}
\left(\M{D}-\M{A}\right)^{-\frac{1}{2}} = 
\left(\M{I}+\frac{1}{2}\M{D}^{-1}\M{A}\right)
\left(\M{D} - \left(\frac{3}{4} \M{D}\cdot(\M{D}^{-1} \M{A})^2 
 + \frac{1}{4}\M{D} \cdot(\M{D}^{-1} \M{A})^3\right)\right)
^{-\frac{1}{2}}.
\end{equation}

Finding inverse square-root factorizations is a key step in
  sampling from Gaussian graphical models:
Given a graphical model of a {\em Gaussian random field}
  specified by its precision matrix $\M{\Lambda}$ and potential vector $\V{h} $,
 i.e., $\text{Pr}(\mathbf{x}| \M{\Lambda}, \V{h}  ) \propto \exp (-\frac{1}{2} \mathbf{x}^{\T} \M{\Lambda} \mathbf{x} + \V{h}^{\T} \mathbf{x})$,  
  efficiently generate {\em i.i.d}
  random samples from this {\em multivariate Gaussian distributions} 
  \cite{LohWainwright}.
If one can compute an efficient sparse representation
  of $\M{C} \approx \M{\Lambda}^{-1/2}$, 
  then one can convert {\em i.i.d.} standard Gaussian random vector 
  $\V{z}$ using $\V{x} =\M{C}\V{z}+\V{\mu}$ (where
  $\V{\mu} = \M{\Lambda}^{-1}\V{h}$) 
  to {i.i.d} random vectors of a Gaussian random field
  that numerically approximates the one defined
   by $(\M{\Lambda},\V{h})$~\cite{ChengCLPT14}.
Furthermore, if the precision matrix $\M{\Lambda} = (\lambda_{i,j})$
  is {\em symmetric diagonally dominant} (SDD), 
  i.e., for all $i$, $\lambda_{i,i} > \sum_{j\neq
  i} \left|\lambda_{i,j}\right|$,
  then one can reduce this factorization problem
  to the problem formulated by Equation (\ref{iter})
  involving an SDDM matrix.

Then, in order to iteratively apply Equation (\ref{iter})
  to build an efficient representation of the inverse square-root factor of $\M{D}-\M{A}$,
  one needs to efficiently construct the second term in Equation (\ref{iter}),
\begin{eqnarray}\label{eqn:Newton}
\M{D} - \left(\frac{3}{4} \M{D}\cdot(\M{D}^{-1} \M{A})^2 
 + \frac{1}{4}\M{D} \cdot(\M{D}^{-1} \M{A})^3\right).
\end{eqnarray}
The quadric and cubic powers in this matrix can be very dense,
  making exact computations involving them expensive.
Instead, we will directly compute an approximation of this matrix that
  still suffices for algorithmic purposes.

Finding an inverse square-root of an SDDM matrix is a special 
  case of the following
   basic algorithmic problem in spectral graph theory and numerical analysis \cite{ChengCLPT14}: 
{\em \begin{quote}
Given an $n\times n$ SDDM matrix $\MM$,
   a non-zero integer $q$, and an approximation parameter $\epsilon$,
   compute an efficient sparse representation of
  an  $n\times n$ linear operator $\tilde{\CC}$ such that 
\[
\M{M}^{1/q} \approx_{\epsilon} \tilde{\CC} \tilde{\CC}^\T
\]
\end{quote}}
\noindent where $\approx_{\epsilon}$ is spectral similarity between linear operators
which we will define at the start of Section~\ref{sec:background}.

The matrix $q^{th}$-root computation appears in several numerical applications 
 and particularly in the analysis of Markov models \cite{HighamLin}.
For example, in his talk for Brain Davies' 65  Birthday conference (2009), 
  Nick Higham quoted an email that he received from a power company
  regarding the usage of an electricity network to illustrate the
  practical needs of taking the $q^{th}$-root of a Markov transition.
{\em 
\begin{quote}
``I have an Excel spreadsheet containing the transition matrix of 
 how a company's 
[Standard \& Poor's] credit 
rating charges from on year to 
 the next. I'd like to be working in
eighths of a year, so the aim is
 to find the {\bf eighth root of the matrix.}''
\end{quote}
}
In our case, note that when the graph 
  is connected, $\M{D}^{-1}\M{A}$ is the transition matrix of a 
  {\em reversible Markov chain} \cite{DavidAldous}, and the first order
  approximation of the $q^{th}$-root transition is
$\M{I} - (\M{I}-\M{D}^{-1}\M{A})^{1/q}$.
Extension of Newton's method then leads to an iterative formula similar 
  to Equation (\ref{iter}) for finding factorization of the $q^{th}$ root.
Thus, the key algorithmic task for obtaining a nearly linear time Newton(-like) algorithm
  for $q^{th}$-root factorizations is the efficient approximation of
  matrix polynomials akin to Equation (\ref{eqn:Newton}).

\subsection{Main Technical Contribution}

We start with a definition that captures the matrix polynomials
  such like Equation (\ref{eqn:Newton}) that arise in the application of
  Newton's or Newton-like methods to graph Laplacians.

\begin{definition}[Random-Walk Matrix-Polynomials]
Let $\M{A}$ and $\M{D}$ be the adjacency matrix and diagonal weighted
 degree matrix of a weighted, undirected graph $G$ respectively.
For a non-negative vector $\V{\alpha} = (\alpha_1,...,\alpha_d)$
  with $\sum_{r=1}^d \alpha_r = 1$, the matrix 
\begin{eqnarray}\label{eqn:mmp}
\M{L}_{\V{\alpha}}(G) = \M{D} - \sum_{r=1}^d \alpha_r\M{D} \cdot 
  \left(\M{D}^{-1}\M{A}\right)^r
\end{eqnarray}
is a $d$-degree {\em random-walk matrix-polynomial} of $G$.
\end{definition}

Random-walk matrix-polynomials naturally include the 
  {\em graph Laplacian} $G$ as the linear case:
when $d = 1$, the matrix polynomial becomes $\M{L}(G)=\M{D} - \M{A}$,
  which is the Laplacian matrix of $G$.
In fact, the following proposition can be established by a simple induction,
  which we prove in Appendix \ref{apx:graphprev}.

\begin{restatable}[Laplacian Preservation]{proposition}{restateGraphPrev}
\label{prop:Laplacian}
\label{lem:graphprev}
For any weighted, undirected graph $G$ with adjacency matrix $\M{A}$ and diagonal matrix 
$\M{D}$, for every non-negative vector 
  $\V{\alpha} = (\alpha_1,...,\alpha_d)$ such that
  with $\sum_{r=1}^d \alpha_r = 1$, 
  the random-walk matrix-polynomial
$\M{L}_{\V{\alpha}}(G)$
  remains a Laplacian matrix.
  \end{restatable}


Consequently, applying spectral sparsification algorithms~\cite{SpielmanTengSpectralSparsification,SpielmanSrivastava,BSS10}
to $\M{L}_{\V{\alpha}}(G)$ gives:

\begin{proposition}[Spectral Sparsifiers of Random-Walk Matrix Polynomials]
For all $G$ and $\V{\alpha}$ as in Proposition \ref{prop:Laplacian},
for any $\epsilon > 0$, there exists a Laplacian matrix 
$\tilde{\M{L}} = {\M{D}} - \tilde{\M{A}}$ with ${O}(n\log n /\epsilon^2)$ non-zeros
  such that  for all $\V{x}\in \Reals{n}$
\begin{eqnarray}\label{eqn:spectralSparsification}
(1-\epsilon)\cdot \V{x}^\T  \tilde{\M{L}} \V{x} \leq 
\V{x}^\T \left(\M{D} - \sum_{r=1}^d \alpha_r\M{D} \cdot 
  \left(\M{D}^{-1}\M{A}\right)^r\right)\V{x}
\leq (1+\epsilon)\cdot \V{x}^\T \tilde{\M{L}} \V{x}.
\end{eqnarray}
\end{proposition}

The Laplacian matrix $\tilde{\M{L}}$ satisfying Equation (\ref{eqn:spectralSparsification})
   is called {\em spectrally similar} with approximation parameter $\epsilon$ to 
  $\M{L}_{\V{\alpha}}(G)$ \cite{SpielmanTengSpectralSparsification}.
The computation of a (nearly) linear size {\em spectral sparsifier} of a dense
  Laplacian matrix is 
  a fundamental algorithmic problem in spectral graph theory
  that has been used in solving linear systems \cite{SpielmanTengLinear,KMP10SDDSolver}
  and combinatorial optimization \cite{CKMST10Maxflow}.
The work of \cite{SpielmanSrivastava,SpielmanTengSpectralSparsification}
  showed that spectral sparsifier of $O(n\log n/\epsilon^{2})$ non-zeros 
  can be constructed in $O(m\log^2 n/\epsilon^{2})$ time 
  for any $n\times n$ Laplacian matrix $\M{L}$ with $m$ non-zeros.
A recent technique of \cite{PengSpielman} can sparsify a degree $2$ random-walk 
matrix-polynomial in nearly linear time.

In this paper, we give the first nearly linear time spectral-sparsification algorithm
  for all random-walk matrix-polynomials.
Our sparsification algorithm is built on 
 the following key mathematical observation that might be interesting on its own:
One can obtain a sharp enough upper bound on the effective resistances
  of the high-order polynomial  $\M{D} - \M{D} (\M{D}^{-1} \M{A})^r$
 from the combination of the linear $\M{D} - \M{A}$ and
  quadratic $\M{D} - \M{A}\M{D}^{-1}\M{A}$ polynomials.

This allows us to design an efficient path sampling algorithm that utilizes this 
  mathematical observation to achieve the critical sparsification.
We prove the following result which generalizes
  the works of \cite{SpielmanTengSpectralSparsification,SpielmanSrivastava,PengSpielman}.

\begin{theorem}[Random-Walk Polynomials Sparsification]
\label{thm:randWalk}
For any weighted, undirected graph $G$ with $n$ vertices and $m$ non-zeros,
 for every non-negative vector $\V{\alpha} = (\alpha_1,...,\alpha_d)$
   with $\sum_{r=1}^d \alpha_r = 1$,
for any $\epsilon > 0$, 
 we can construct in time  
$O(d^2\cdot m\cdot \log^2 n/\epsilon^{2})$
  a spectral sparsifier with $O(n\log n/\epsilon^{2})$ non-zeros and
  approximation parameter $\epsilon$
  for the random-walk matrix-polynomial
$\M{L}_{\V{\alpha}}(G)$.
\end{theorem}

The total work of our sparsification algorithm depends
  quadratically in the degree of the polynomial,
  which could be expensive when $d=\Theta(n^c)$ for some constant $c > 0$.
In Section \ref{sec:higherDegree} we present, for even degrees $d$,
  a more efficient algorithm to sparsify
$\M{D} - \M{D}(\M{D}^{-1}\M{A})^d$.
We will show that, for any positive integer $r$, if we are given $\tilde{\AA}$ such that
\begin{align}
  \DD-\tilde{\AA}
  \approx
  \DD-\DD \left(\DD^{-1}\AA\right)^{2r},
\end{align}
we can construct a sparse matrix $\tilde{\AA}_{\times}$ and a sparse matrix 
 $\tilde{\AA}_{+}$ such that 
\begin{align}
  \DD - \tilde{\AA}_{\times} 
  \approx
  \DD-\DD \left(\DD^{-1}\AA\right)^{4r} \quad \mbox{and} \quad
   \DD - \tilde{\AA}_{+}
  \approx
  \DD-\DD \left(\DD^{-1}\AA\right)^{2r+4}.
\end{align}


Applying these two routines inductively gives an algorithm that,
  for any $d$ divisible by $4$,
  approximates the $d$-degree random-walk matrix-monomial
  in time polylogarithmic in $d$.
We can also extend this algorithm to handle all the even-degree monomials.
Because when $d = \OO{1/\epsilon}$, we can directly invoke Theorem \ref{thm:randWalk}
  to sparsify $\M{D} - \M{D} (\M{D}^{-1} \M{A})^{d}$,
  and when $d = \Omega(1/\epsilon)$, we know that 
  $\M{D} - \M{D} (\M{D}^{-1} \M{A})^{d} \approx_\epsilon
  \M{D} - \M{D} (\M{D}^{-1} \M{A})^{d+2}$,
  therefore any even-degree monomial can be replaced by a degree $4r$ monomial,
  while introducing a small error only.

\begin{theorem}[High Degree]
\label{thm:power}
For any even integer $d$, 
  let $\M{L}_{G_{d}} = \M{D}- \M{D} (\M{D}^{-1} \M{A})^{d}$ be the $d$-step random walk matrix,
For any $\epsilon > 0$, we can construct 
  a graph Laplacian $\M{L}_{\tilde G}$ 
  with $O(n\log n/\epsilon^{2})$ nonzero entries,
  in total work $O(m \cdot \log^3 n\cdot \log^5 (d) /\epsilon^4)$,
  such that
$
  \M{L}_{\tilde G} \approx_\epsilon \M{L}_{G_d}.
$
\end{theorem}

While we build our construction on the earlier work 
 \cite{SpielmanTengSpectralSparsification,SpielmanSrivastava,PengSpielman}
  for graph sparsification,
we need to overcome some significant difficulties
  posed by high degree matrix polynomials,
which appear be algorithmically challenging:
  The matrix $(\M{D}^{-1}\M{A})^r$ is  
  defined by all paths of length $r$,
  whose precise calculation would be prohibitively expensive
  due to the cost of matrix multiplication and densification in the matrix powers.
Moreover, the algorithm of \cite{PengSpielman} relies an explicit clique-like representation
 of edges in the quadratic power and expanders, which is much more specialized.

\subsection{Some Applications}

Matrix polynomials are involved in numerical methods such as Newton's or Newton-like
  methods, which have been widely 
used for finding solutions of matrix equations.
Our sparsification algorithm can be immediately applied to speed up 
  these numerical methods that involves SDDM matrices.
For example, in the application of finding the inverse square root of an SDDM matrix
  $\M{D}-\M{A}$,
  we can directly sparsify the cubic
  matrix polynomial given in Equation \ref{eqn:Newton}, and iteratively
  approximate the inverse-square root factor of $\M{D} - \M{A}$ using 
  Equation \ref{iter}.
This leads to a simpler and faster algorithm than the
  one presented in \cite{ChengCLPT14},
  which circumvents the challenging problem of sparsifying high-degree random-walk 
  matrix polynomials at the cost of slower convergences and complex approximation.
The simplicity of the new algorithm comes from the fact that we no longer need
  the Maclaurin series for conditioning the numerical iterations.
The convergence analysis of the new algorithm follows 
  the standard analysis of the Newton's method, with careful
  adaptation to handle the approximation errors introduced by the
  spectral sparsification. 
The elimination of the Maclaurin series speeds up the previous algorithm by 
  a factor of $\log\log \kappa$, where $\kappa$ is the relative condition
  number of $\M{D}-\M{A}$.
  
In general, our sparsification algorithm can be used inside
 the Newton-like method for approximating the inverse $q^{th}$-root 
  of SDDM matrices 
  to obtain a simpler and faster
  nearly linear time algorithm than the one presented in 
  \cite{ChengCLPT14} for $q\in \mathbb{Z}_+$, with reduction formula as follows
  \begin{align}
    \left(\M{I} - \M{X} \right)^{-1/q}
    =
    \left(\M{I} + \frac{\M{X}}{2q} \right)
    \left[ \left(\M{I} + \frac{\M{X}}{2q}\right)^{2q} \left(\M{I} - \M{X} \right)  \right]^{-1/q}
    \left(\M{I} + \frac{\M{X}}{2q} \right).
  \end{align}
Our mathematical and algorithmic advances enable the sparsification of the $(2q+1)$-degree polynomials in the middle, in turn speed up the previous algorithm by a factor of
 $\log(\log (\kappa)/\epsilon)$.

By Proposition \ref{lem:graphprev}, the random-walk matrix polynomial
$\M{L}_{\V{\alpha}}(G) = \M{D} - \sum_{r=1}^d \alpha_r\M{D} \cdot 
  \left(\M{D}^{-1}\M{A}\right)^r
$ 
defines a weighted graph $G_{\V{\alpha}}$ whose adjacency matrix is 
 $ \sum_{r=1}^d \alpha_r\M{D} \cdot 
  \left(\M{D}^{-1}\M{A}\right)^r$, and overlays $d$ graphs induced by 
  the multi-step random walks.
While  $\M{D} - \M{A}$ and $\M{D} - \M{A}\M{D}^{-1}\M{A}$ offers
  a good enough bound on the effective resistances of edges in $G_{\V{\alpha}}$
 for the purpose of path sampling,
  these estimates are relatively loose comparing with the standard approximation condition.
Because spectral similarity implies effective-resistance similarity~\cite{SpielmanTengSpectralSparsification,SpielmanSrivastava},
our sparsification algorithm together with 
  the construction of Spielman and Srivastava~\cite{SpielmanSrivastava}
  provide an efficient data structure for effective resistances in $G_{\V{\alpha}}$.
After nearly linear preprocessing time, we can answer queries regarding
approximate effective resistances in $G_{\V{\alpha}}$ in logarithmic time.



Due to these connections to widely used tools in numerical
and network analysis, we anticipate our nearly linear time
sparsification algorithm could be useful for a variety of other tasks.

\section{Background and Notation}
\label{sec:background}
  
We assume $G = (V, E, w)$ 
  is a weighted undirected graph with $n = |V|$ vertices,
  $m = |E|$ edges and edge weights $w_e > 0$.
Let $\M{A} = (a_{i,j})$ denote the adjacency matrix of $G$, i.e., 
  $a_{i,j} = w(i,j)$.
We let $\M{D}$ to denote the diagonal matrix containing 
  weighted degrees of vertices.
Note that $\M{D}^{-1}\M{A}$ is the {\em transition matrix} of
  random walks on the graph and $\M{L}_G = \M{D} - \M{A}$ 
  is the {\em Laplacian matrix} of $G$.
It is well known that for any vector $\V{x} = (x_1,...,x_n)$, 
\begin{align}
  \V{x}^\T \M{L}_G \V{x} = \sum_{(u,v)\in E} (x_u - x_v)^2 w_{uv}
\end{align}

We use $G_r$ to denote the graph introduced by $r$-step random walks on $G$.
We have  $\M{L}_{G_2} = \M{D} - \M{A} \M{D}^{-1} \M{A}$,
  and $\M{L}_{G_r} = \M{D} - \M{D} (\M{D}^{-1} \M{A})^r$ in general.

In our analysis, we will make extensive use of spectral approximations
based on the Loewner partial ordering of positive semidefinite matrices.
Given two matrices $\M{X}$ and $\M{Y}$, we use 
$ \M{Y} \succcurlyeq \M{X}$ (or equivalently $ \M{X} \preccurlyeq \M{Y}$)
to denote that $\M{Y}-\M{X}$ is positive semi-definite.
Approximations using this ordering obey usual intuitions with
approximations of positive scalars.
We will also use a compressed, symmetric notation in situations where
we have mirroring upper and lower bounds.
We say $\M{X} \approx_\epsilon \M{Y}$ when
\begin{align}
\exp \left(\epsilon\right) \M{X} \succcurlyeq \M{Y} \succcurlyeq  \exp \left( -\epsilon \right) \M{X},
\end{align}
 We use the following standard facts about this notion of approximation.
 
\begin{fact}\label{fact}
For positive semi-definite matrices
$\M{X}$, $\M{Y}$, $\M{W}$ and $\M{Z}$,
\begin{enumerate}
\item if $\M{Y} \approx_{\epsilon} \M{Z}$,
then $\M{X} + \M{Y} \approx_{\epsilon} \M{X} + \M{Z}$;

\item \label{fact:sum}
if $\M{X} \approx_{\epsilon} \M{Y}$ and $\M{W} \approx_{\epsilon} \M{Z}$,
  then $\M{X} + \M{W} \approx_{\epsilon} \M{Y} + \M{Z} $;

\item \label{fact:transit}
if $\M{X} \approx_{\epsilon_1} \M{Y}$ and
$\M{Y} \approx_{\epsilon_2} \M{Z}$,
then $\M{X} \approx_{\epsilon_1 + \epsilon_2} \M{Z}$;

\item
if $\M{X}$ and $\M{Y}$ are positive definite matrices
  such that $\M{X} \approx_{\epsilon} \M{Y}$,
  then $\M{X}^{-1} \approx_{\epsilon} \M{Y}^{-1}$;

\item \label{fact:compose}
if $\M{X} \approx_{\epsilon} \M{Y}$
and $\M{V}$ is a matrix, then
$  \M{V}^{\T} \M{X} \M{V} \approx_{\epsilon} \M{V}^{\T} \M{Y} \M{V}.$

\end{enumerate}
\end{fact}


The Laplacian matrix is closely related to electrical flow 
  \cite{SpielmanSrivastava,CKMST10Maxflow}.
For an edge with weight $w(e)$, we view it as a resistor with resistance $r(e) = 1/w(e)$.
Recall that the {\em effective resistance} between two vertices $u$ and $v$
  $R(u,v)$ is defined as the potential difference induced between them
  when a unit current is injected at one and extracted at the other.
Let $\V{e}_i$ denote the vector with 1 in the $i$-th entry and 0 everywhere else,
  the effective resistance $R(u,v)$ equals to
  $(\V{e}_u - \V{e}_v)^\T \M{L}^{\dagger} (\V{e}_u - \V{e}_v)$,
  where $\M{L}^{\dagger}$ is the {\em Moore-Penrose Pseudoinverse} of $\M{L}$.
From this expression, we can see that effective resistance obeys triangle inequality.
Also note that adding edges to a graph does not increase the
  effective resistance between any pair of nodes.

\begin{figure}[t]

  \begin{algbox}
  $\tilde{G}= \algname(G=\{V,E,w\},\tau_e,M)$
  \begin{enumerate}
  \item Initialize graph $\tilde{G}=\{V,\emptyset\}$.
  \item For $i$ from $1$ to $M$: 
  \label{alg:sampleStep}

  Sample an edge $e$ from $E$ with $p_e = {\tau_e} / (\sum_{e\in E} \tau_e)$.
  Add $e$ to $\tilde{G}$ with weight ${w(e)} / (M \tau_e)$.
  \item Return graph $\tilde{G}$.
  \end{enumerate}
  \end{algbox}
  
  \caption{Pseudocode for Sampling by Effective Resistances}
  \label{fig:ERSparsify}
  \end{figure}


  

In \cite{SpielmanSrivastava}, it was shown that
 oversampling the edges using upper bound on the effective resistance
 suffices for constructing spectral sparsifiers.
The theoretical guarantees for this sampling process were strengthened
in \cite{kelner2013spectral}.
A pseudocode of this algorithm is given in Figure~\ref{fig:ERSparsify}, and its guarantees
can be stated as follows.

\begin{theorem}
\label{thm:ERsparsify}
Given a weighted undirected graph $G$,
  and upper bound on its effective resistance $Z(e) \ge R(e)$.
For any approximation parameter $\epsilon>0$,
  there exists $M = \OO{ \log n / \epsilon^{2} \cdot (\sum_{e\in E} \tau_e)}$,
  with $\tau_e = w(e) Z(e)$,
  such that with probability at least $1-\frac{1}{n}$,
  $\tilde{G} = \algname(G,w,\tau_e,M)$
  has at most $M$ edges, and satisfies
  \begin{align}
    (1-\epsilon) \M{L}_{G} \preccurlyeq \M{L}_{\tilde{G}} \preccurlyeq (1+\epsilon) \M{L}_{{G}}.
  \end{align}
\end{theorem} 

The equivalence between solving linear systems in such matrices
and graph Laplacians, weakly-SDD matrices, and M-matrices
are well known~\cite{SpielmanTengSolver,DaitchS08,Kelner,PengSpielman}.


\section{Sparsification by Sampling Random Walks}
\label{sec:sampling}

We sparsify $\M{L}_{G_r} = \M{D} - \M{D} (\M{D}^{-1} \M{A})^r$ in two steps.
In the first and critical step, we obtain an initial sparsifier
  with $O(dm\log n/\epsilon^2)$ non-zeros for 
  $\M{L}_{G_r}$
  using an upper bound estimate on the effective resistance of $G_r$
  obtained from 
  $\M{L}_{G} = \M{D} - \M{A}$ and $\M{L}_{G_2} =\M{D} - \M{A}\M{D}^{-1}\M{A}$.
In the second step, we apply the standard spectral sparsification
  algorithms to further reduce the number of nonzeros to $O(n\log n/\epsilon^2)$.

In the first step sparsification, we bound the effective resistance on $G_r$
  using Lemma \ref{lem:support} (proof included in Appendix \ref{apx:support}),
  which allows us to use the resistance of any length-$r$ path on $G$
  to upper bound the effective resistance between its two endpoints on $G_r$.
Lemma \ref{lem:upperboundEffRes} shows that we can sample by these estimates efficiently.

\begin{restatable}[Two-Step Supports]{lemma}{restateSupportLowOrder}
  \label{lem:support}
  For a graph Laplacian matrix $\M{L} = \M{D} - \M{A}$, with
  diagonal matrix $\M{D}$ and nonnegative off-diagonal $\M{A}$,
  for all positive odd integer $r$, we have
  \begin{align}
  \Half \M{L}_G \preceq \M{L}_{G_r} \preceq r \M{L}_G.
  \end{align}
  and for all positive even integers $r$ we have
  \begin{align}
  \M{L}_{G_2} \preceq \M{L}_{G_r} \preceq \frac{r}{2} \M{L}_{G_2}.
  \end{align}
\end{restatable}



For each length-$r$ path $\V{p} = (u_0 \ldots u_r)$ in $G$,
  we have a corresponding edge in $G_r$,
  with weight proportional to the chance of this particular
  path showing up in the random walk.
We can view $G_r$ as the union of these edges, i.e.,
  $\M{L}_{G_r}(u_0, u_r) = \sum_{\V{p}=(u_0 \ldots u_r)} w(\V{p})$.

We bound the effective resistance on $G_r$ in two different ways.
If $r$ is odd, $G$ is a good approximation of $G_r$,
  so we can obtain an upper bound using the resistance
  of a length-$r$ (not necessarily simple) path on $G$.
If $r$ is even, $G_2$ is a good approximation of $G_r$, in this case,
  we get an upper bound by composing the effective resistance of 2-hop paths
  in different subgraphs of $G_2$.

\begin{lemma}[Upper Bound on Effective Resistance]
  \label{lem:upperboundEffRes}
For a graph $G$ with Laplacian matrix $\M{L} = \M{D} - \M{A}$,
let $\M{L}_{G_r} = \M{D} - \M{D} (\M{D}^{-1} \M{A})^r$
  be its $r$-step random-walk matrix.
Then, the effective resistance between two vertices $u_0$ and $u_r$ on
  $\M{L}_{G_r}$ is upper bounded by
  \begin{align}
  R_{G_r} (u_0,u_r) \leq \sum_{i=1}^r \frac{2}{\M{A}({u_{i-1},u_{i}})},
  \end{align}
  where $(u_0 \ldots u_r)$ is a path in $G$.
  \end{lemma}

\begin{proof}
When $r$ is a positive odd integer, by Lemma \ref{lem:support}, we have that
  $\Half \M{L}_G \preceq \M{L}_{G_r}$,
  which implies that for any edge $(u,v)$ in $G$, 
  \begin{align}
  R_{G_r} (u,v) \leq 2 r_G(u,v) = \frac{2}{\M{A} ({u,v})}.
  \end{align}
Because effective resistance satisfies triangular inequality,
  this concludes the proof for odd $r$.

When $r$ is even, by Lemma \ref{lem:support}, we have that
  $\M{L}_{G_2} \preceq \M{L}_{G_r}$.
The effective resistance of $\M{D}-\M{A}\M{D}^{-1}\M{A}$
  is studied in \cite{PengSpielman} 
   and restated in Appendix \ref{apx:supportRankOne}.
For the subgraph $G_2(u)$ anchored at the vertex $u$
  (the subgraph formed by length-2 paths where the middle node is $u$),
  for any two of its neighbors $v_1$ and $v_2$, we have
  \begin{align}
  R_{G_r} (v_1, v_2) \le R_{G_2(u)}(v_1, v_2) = \frac{1}{\M{A} ({v_1,u})} + \frac{1}{\M{A} ({u,v_2})}.
  \end{align}
Because we have the above upper bound for any 2-hop path $(v_1, u, v_2)$,
  by the triangular inequality of effective resistance,
  the lemma holds for even $r$ as well.
\end{proof}

Now we could sparsify $G_r$ by sampling random walks
  according to the approximate effective resistance.
The following mathematical identity will be crucial in our analysis.

\begin{lemma}[A Random-Walk Identity]
Given a graph $G$ with $m$ edges, consider the $r$-step
  random-walk graph $G_r$ and the corresponding Laplacian matrix
  $\M{L}_{G_r} = \M{D} - \M{D} (\M{D}^{-1} \M{A})^r$.
For a length-$r$ path $\V{p} = (u_0 \ldots u_r)$ on $G$, we have
  \begin{align*}
  w(\V{p}) = \frac{\prod_{i=1}^{r} \M{A} ({u_{i-1},u_{i}})}{\prod_{i=1}^{r-1} \DD (u_i,u_i)},
  \quad
  Z(\V{p}) = \sum_{i=1}^r \frac{2}{A (u_{i-1},u_i)}.
  \end{align*}
The summation of $w(\V{p}) Z(\V{p})$ over all length-$r$ paths satisfies
  \begin{align}
  \sum_{\V{p}} w(\V{p}) \cdot Z(\V{p}) = 2 r m.
  \end{align}
\end{lemma}

\begin{proof}
We substitute the expression for $w(\V{p})$ and $Z(\V{p})$ in to the summation.
  \begin{align}
  \sum_{\V{p} = (u_0 \ldots u_r)} w(\V{p}) Z(\V{p})
  &= \sum_{\V{p}} \left(\sum_{i=1}^r \frac{2}{\AA (u_{i-1},u_i)}\right) \left(\frac{\prod_{j=1}^{r} \M{A} ({u_{j-1},u_{j}})}{\prod_{j=1}^{r-1} \DD (u_j,u_j)}\right) \\
  &= 2 \sum_{\V{p}} \sum_{i=1}^r \left(\frac{\prod_{j=1}^{i-1} \M{A} ({u_{j-1},u_{j}}) \prod_{j=i}^{r-1} \M{A} ({u_{j},u_{j+1}})}{\prod_{j=1}^{r-1} \DD (u_j,u_j)} \right) \label{eqn:pathEdgeSum}\\
  &= 2 \sum_{e \in G} \sum_{i=1}^r \left( \sum_{\V{p} \text{~with~} (u_{i-1},u_i)=e} \frac{\prod_{j=1}^{i-1} \M{A} ({u_{j-1},u_{j}}) \prod_{j=i}^{r-1} \M{A} ({u_{j},u_{j+1}})}{\prod_{j=1}^{r-1} \DD (u_j,u_j)} \right) \label{eqn:edgePathSum}\\
  &= 2 \sum_{e \in G} \sum_{i=1}^r \left( \sum_{\V{p} \text{~with~} (u_{i-1},u_i)=e} \frac{\prod_{j=1}^{i-1} \M{A} ({u_{j-1},u_{j}})}{\prod_{j=1}^{i-1} \DD (u_j,u_j)} \cdot \frac{\prod_{j=i}^{r-1} \M{A} ({u_{j},u_{j+1}})}{\prod_{j=i}^{r-1} \DD (u_{j},u_{j})} \right)  \label{eqn:beforeCancel} \\
  &= 2 \sum_{e \in G} \sum_{i=1}^r 1 \label{eqn:afterCancel} \\
  &= 2 m r.
  \end{align}

From Equation \ref{eqn:pathEdgeSum} to \ref{eqn:edgePathSum},
  instead of enumerating all paths,
  we first fix an edge $e$ to be the $i$-th edge on the path,
  and then extend from both ends of $e$.
From Equation \ref{eqn:beforeCancel} to \ref{eqn:afterCancel},
  we sum over indices iteratively from $u_0$ to $u_{i-1}$,
  and from $u_r$ to $u_i$.
Because 
$\M{D}(u,u) = \sum_{v} \M{A} (u,v)$,
  this summation over all possible paths
  anchored at $(u_{i-1}, u_{i})$ equals to 1.
  \end{proof}

Now we show that we can perform Step (\ref{alg:sampleStep}) in $\algname$ efficiently.
We take samples in the same way we cancel the terms in the previous proof.
Recall that sampling an edge from $G_r$ corresponds to sampling a path of length $r$ in $G$.
\begin{algbox}
  Sample a path $\V{p}$ from $G$ with probability proportional to $\tau_{\V{p}} = w(\V{p}) Z(\V{p})$:
  \begin{enumerate}
  \item [a.] Pick an integer $k \in [1:r]$ and an edge $e \in G$, both uniformly at random.
  \item [b.] Perform $(k-1)$-step random walk from one end of $e$.
  \item [c.] Perform $(r-k)$-step random walk from the other end of $e$.
  \item [d.] Keep track of $w(\V{p})$ during the process, and finally add a fraction of this edge to our sparsifier.
  \end{enumerate}
\end{algbox}

  \begin{lemma}
  \label{lem:implicitSample}
  There exists an algorithm for the Step (\ref{alg:sampleStep}) in $\algname$, such that after preprocessing with work $\OO{n}$, it can draw an edge $e$ as in Step (\ref{alg:sampleStep}) with work $\OO{r \cdot \log n}$. 
  \end{lemma}
  \begin{proof}
The task is to draw a sample $\V{p} = (u_0 \ldots u_r)$ from the multivariate distribution $\mathcal{D}$  
  \begin{align}
  \Pr(u_0 \ldots u_r) 
  = \frac{1}{2rm} 
  \cdot
  \left( \sum_{i=1}^r \frac{2}{\M{A} ({u_{i-1},u_{i}})}\right)
  \cdot
  \left(
  \frac{\prod_{i=1}^{r} \M{A} ({u_{i-1},u_{i}})}{\prod_{j=1}^{r-1} \M{D} (u_i,u_i)}
  \right).
  \end{align}
For any fixed $k \in [1:r]$ and $e \in G$, we can rewrite the distribution as
  \begin{align}
  \begin{split}
  \label{eqn:sampleWalk}
  \Pr(u_0 \ldots u_r)
  &= \Pr((u_{k-1},u_k) = e) \cdot \Pr(u_0 \ldots u_{k-2}, u_{k+1} \ldots u_r \vert (u_{k-1},u_k) = e)   \\
  &= \Pr((u_{k-1},u_k) = e) \cdot \Pr(u_0 \ldots u_{k-2} \vert u_{k-1}) \cdot \Pr(u_r \ldots u_{k+1} \vert u_{k})  \\
  &= \Pr((u_{k-1},u_k) = e) \cdot \prod_{1}^{i=k-1} \Pr(u_{i-1} \vert u_{i}) \cdot \prod_{i=k+1}^{r} \Pr(u_{i} \vert u_{i-1})
  \end{split}
  \end{align}
Note that $\Pr((u_{k-1},u_k) = e) = \frac{1}{m}$,
  and $\Pr(u_{i-1} \vert u_i) = \M{A}(u_i,u_{i-1}) / \M{D}(u_i,u_i)$.
The three terms in Equation \ref{eqn:sampleWalk}
  corresponds to Step (a)-(c) in the sampling algorithm stated above.
With linear preprocessing time, we can draw an uniform random edge in time $\OO{\log n}$,
  and we can also simulate two random walks with total length $r$ in time $\OO{r \log n}$,
  so Step (\ref{alg:sampleStep}) in $\algname$ can be done within $\OO{r \log n}$ time.
  \end{proof}

Combining this with spectral sparsifiers gives our algorithm for efficiently
sparsifying low-degree polynomials.

\begin{proof}[Proof of Theorem \ref{thm:randWalk}]
First we show on how to sparsify a degree-$d$ monomial 
  $\M{L}_{G_d} = \M{D} - \M{D} \cdot \left(\M{D}^{-1}\M{A}\right)^d$.
We use the sampling algorithm described in Theorem \ref{thm:ERsparsify},
  together with upper bounds on effective resistance of $G_d$
  obtained from $\M{D} - \M{A}$ and $\M{D} - \M{A} \M{D}^{-1} \M{A}$.
The total number of samples requires is ${O}(d m \log n /\epsilon^{2})$.
We use Lemma \ref{lem:implicitSample} to draw a single edge from $G_r$,
  where we sample $d$-step random walks on $\M{D} - \M{A}$,
  so the total running time is ${O}(d^2 m \log^2 n/\epsilon^{2})$.
Now that we have a spectral sparsifier with ${O}(d m \log n /\epsilon^{2})$ edges,
  we can sparsify one more time to reduce the number of edges to $O(n\log n/\epsilon^{2})$ by 
  \cite{SpielmanTengSpectralSparsification,SpielmanSrivastava}
   in time ${O}(d m \log^2 n /\epsilon^{2})$.

To sparsify a random-walk matrix-polynomial $\M{L}_{\V{\alpha}}(G)$,
  we sparsify all the even/odd terms together,
  so the upper bound of effective resistance in Lemma \ref{lem:upperboundEffRes} still holds.
To sample an edge, we first decide the length $r$ of the path,
  according to the probability distribution 
  $\text{Pr}(\text{length}=r|r\text{ is odd/even}) \propto \alpha_r$.
\end{proof}

\section{Sparsification of Higher Degree Matrix-Monomials}
\label{sec:higherDegree}
We now put together the components of the algorithm for
sparsifying higher degree monomials.
Given a positive even integer $2r$ and $\tilde{\AA}$ such that
  $\DD - \tilde{\AA} \approx_{\epsilon} \DD - \DD(\DD^{-1} \AA)^{2r}$,
  we will show that we can efficiently compute
\begin{enumerate}
\item $\tilde{\AA}_{\times}$ such that $\DD - \tilde{\AA}_{\times}
  \approx_{\epsilon'} \DD - \tilde{\AA} \DD^{-1} \tilde{\AA}\approx_{\epsilon}\M{L}_{G_{4r}}$
  in Subsection~\ref{subsec:times}, and
\item $\tilde{\AA}_{+}$ such that $\DD - \tilde{\AA}_{+}
  \approx_{\epsilon'} \DD - (\M{A} \M{D}^{-1})^2 \tilde{\AA} (\M{D}^{-1} \M{A})^2 \approx_{\epsilon}\M{L}_{G_{2r+4}}$
  in Subsection~\ref{subsec:plus}.
\end{enumerate}
Putting these together then gives our main theorem about sparsifying
high degree monomials from Theorem~\ref{thm:power}.

\subsection{Constructing $\tilde{\AA}_{\times}$}
\label{subsec:times}


We construct $\tilde{\AA}_{\times}$ by sparsifying $\DD - \tilde{\AA} \DD^{-1} \tilde{\AA}$ with the algorithm described in Section \ref{sec:sampling}.
 The remaining is to show that spectral approximation holds under squaring, i.e.,
\begin{align}
  \DD - \tilde{\AA} 
  & \approx_{\epsilon'}
  \DD-\DD \left(\DD^{-1}\AA\right)^{2r} \\
  \Rightarrow
  \DD - \tilde{\AA} \DD^{-1} \tilde{\AA}
  & \approx_{\epsilon'}
  \DD-\DD \left(\DD^{-1}\AA\right)^{4r},
\end{align}
which directly follows 
Lemma \ref{lem:lemmaDplusA}
 and Lemma \ref{lem:powerupSpsfPrsv}. The result is summarized in the following lemma.

\begin{lemma}\label{lem:matrixCompSquare}
Let the graph Laplacian $\M{L}_{G}=\DD-\AA$ and $\M{L}_{\tilde{G}_{2r}} = \DD-\tilde{\AA}$ for $r \in \mathbb{Z}_+$, such that (1) $\DD-\tilde{\AA} \approx_{\epsilon'}\DD- \DD\left(\DD^{-1}\AA\right)^{2r}$, and (2) $\tilde{\AA}$ contains $m$ nonzero entries,
 then for any $\epsilon > 0$,
  we can construct 
  in total work $\OO{m \log^{3} n /\epsilon^4}$,
  a graph Laplacian $\M{L}_{\tilde{G}_{\times}} = \DD-\tilde{\AA}_{\times}$  with $\tilde{\AA}_{\times}$ containing at most $\OO{n \log n /\epsilon^2 }$ nonzero entries, such that 
\begin{align}
\DD-\tilde{\AA}_{\times}
    \approx_{\epsilon'+\epsilon}
  \DD-\DD \left(\DD^{-1}\AA\right)^{4r}.
\end{align}
\end{lemma}

First, we will start with the fact that $\M{D}  - \M{A}\M{D}^{-1}\AA$ is the
Schur complement of the matrix
\[
\left[
\begin{array}{cc}
\M{D} & -\M{A}\\
-\M{A}& \M{D}
\end{array}
\right],
\]
onto the second half of the vertices.

This is to allow the use of the following Lemma regarding
Schur complement:
\begin{fact}[Lemma B.1. from~\cite{MillerP13}]
\label{lem:schurApprox}
Suppose $\M{M}$ and $\tilde{\M{M}}$ are positive semi-definite
matrices satisfying $\M{M} \approx_{\epsilon} \tilde{\M{M}}$,
then their Schur complements on the same set of vertices also satisfy
$\M{M}_{\text{schur}} \approx_{\epsilon} \tilde{\M{M}}_{\text{schur}}$.
\end{fact}

We also need the following facts about even random walks.
\begin{lemma}\label{lem:lemmaDplusA}
If $0 \preccurlyeq \M{A}$ and 
\begin{align}\label{eq:lemmaDplusA}
	\left(1-\epsilon\right)\left(\M{D} - \M{A}\right)
 \preccurlyeq 
   \M{D} - \tilde{\M{A}}
 \preccurlyeq  
 \left(1+\epsilon\right)\left(\M{D} - \M{A}\right),
\end{align}
 then
\begin{align}
\left(1-\epsilon\right)\left(\M{D} + \M{A}\right)
 \preccurlyeq 
   \M{D} + \tilde{\M{A}}
 \preccurlyeq  
 \left(1+\epsilon\right)\left(\M{D} + \M{A}\right).
\end{align}
\end{lemma}
\begin{proof}

Rearranging the leftmost condition in Equation \ref{eq:lemmaDplusA} gives
\begin{align}
\tilde{\M{A}} \preccurlyeq \epsilon \M{D} + \left(1 -\epsilon\right)   \M{A}
\end{align}
Adding $\M{D}$ to both sides then gives
\begin{align}
\M{D} + \tilde{\M{A}}
\preccurlyeq \left(1+ \epsilon \right) \M{D}
+ \left(1 -\epsilon\right)    \M{A}.
\end{align}
Combining with $0 \preccurlyeq \AA$ gives
\begin{align}
\M{D} + \tilde{\M{A}}
\preccurlyeq \left(1+ \epsilon \right) \left( \M{D} + \M{A} \right).
\end{align}
Similarly, we can prove 
\begin{align}
	\left(1- \epsilon \right) \left( \M{D} + \M{A} \right)
    \preccurlyeq \M{D} + \tilde{\M{A}}.
\end{align}
\end{proof}

\begin{lemma}\label{lem:powerupSpsfPrsv}
If $\M{D} - \M{A} \approx_{\epsilon} \M{D} - \tilde{\M{A}}$ and $\M{D} + \M{A} \approx_{\epsilon} \M{D} + \tilde{\M{A}}$, then
\begin{align}
\M{D} - \M{A}\DD^{-1}\M{A} \approx_{\epsilon} \M{D} - \tilde{\M{A}}\DD^{-1}\tilde{\M{A}}.
\end{align}
\end{lemma}
\begin{proof}
Because of Lemma~\ref{lem:schurApprox}, it suffices to show that
\begin{align}\label{eq:schurApprox}
\left[
\begin{array}{cc}
\M{D} & -\M{A}\\
-\M{A}& \M{D}
\end{array}
\right]
\approx_{\epsilon}
\left[
\begin{array}{cc}
\M{D} & - \tilde{\M{A}}\\
- \tilde{\M{A}} & \M{D}
\end{array}
\right].
\end{align}

Consider a test vector
\[
\left[
\begin{array}{c}
\V{x}\\
\V{y}\\
\end{array}
\right],
\]
we have 
\begin{align*}
\left[
\begin{array}{cc}
\V{x}^\T &
\V{y}^\T \\
\end{array}
\right]
	\left[
\begin{array}{cc}
\M{D} & -\M{A}\\
-\M{A}& \M{D}
\end{array}
\right]
&
\left[
\begin{array}{c}
\V{x}\\
\V{y}\\
\end{array}
\right] \\
& =
\Half 
\left[
\left(\V{x}+\V{y}\right)^\T \left(\DD-\AA \right)\left(\V{x}+\V{y}\right)^\T  
+
\left(\V{x}-\V{y}\right)^\T \left(\DD+\AA \right)\left(\V{x}-\V{y}\right)^\T 
\right], \numberthis
\end{align*}
which leads to Equation \ref{eq:schurApprox}
, since we have $\M{D} - \M{A} \approx_{\epsilon} \M{D} - \tilde{\M{A}}$ and $\M{D} + \M{A} \approx_{\epsilon} \M{D} + \tilde{\M{A}}$.
\end{proof}

\subsection{Constructing $\tilde{\AA}_{+}$}
\label{subsec:plus}

We can then extend this squaring routine by composing its walk with
smaller random walks on each side. That is, we sparsify $\DD - \left( \M{A} \M{D}^{-1} \right)^2 \tilde{\M{A}} \left( \M{D}^{-1} \M{A} \right)^2$ with the support of $\DD-\AA \DD^{-1} \AA$ and $\DD - \tilde{\AA}$. First, we need to prove that the symmetric composition preserve the spectral approximation.

\begin{lemma}\label{lem:symComposition}
If $\M{D} - {\M{A}} \approx_{\epsilon} \M{D} - \tilde{\M{A}}$,
then for any symmetric $\widehat{\M{A}}$ such that
$0 \preccurlyeq \widehat{\M{A}} \preccurlyeq \M{D}$, we have
\[
\M{D} -  \widehat{\M{A}} \M{D}^{-1} \M{A} \M{D}^{-1}  \widehat{\M{A}}
\approx_{\epsilon} \M{D} -   \widehat{\M{A}} \M{D}^{-1} \tilde{\M{A}} \M{D}^{-1} \widehat{\M{A}}.
\]
\end{lemma}

\begin{proof}
Since $\widehat{\M{A}} \M{D}^{-1} = ( \M{D}^{-1} \widehat{\M{A}})^\T$,
we can compose the given identity with $\M{D}^{-1} \widehat{\M{A}}$
using Fact~\ref{fact}.\ref{fact:compose} to give:
\[
\left( \widehat{\M{A}} \M{D}^{-1}  \right) \left( \M{D} - \M{A} \right) \left( \M{D}^{-1} \widehat{\M{A}} \right)
\approx_{\epsilon}
\left( \widehat{\M{A}} \M{D}^{-1}  \right) \left( \M{D} - \tilde{\M{A}} \right) \left( \M{D}^{-1} \widehat{\M{A}} \right),
\]
or if expanded:
\[
\widehat{\M{A}} \M{D}^{-1} \widehat{\M{A}}
- \widehat{\M{A}} \M{D}^{-1} \M{A} \M{D}^{-1}  \widehat{\M{A}}
\approx_{\epsilon}
\widehat{\M{A}} \M{D}^{-1} \widehat{\M{A}}
-\widehat{\M{A}} \M{D}^{-1} \tilde{\M{A}} \M{D}^{-1}  \widehat{\M{A}}
\]

Also, since $\M{D} - \widehat{\M{A}} \succcurlyeq 0$, we have
\[
\M{D} - \widehat{\M{A}} \M{D}^{-1} \widehat{\M{A}} \succcurlyeq 0,
\]
which allows us to write an approximation of the form
\[
\M{D} - \widehat{\M{A}} \M{D}^{-1} \widehat{\M{A}} \approx_{\epsilon}
 \M{D} - \widehat{\M{A}} \M{D}^{-1} \widehat{\M{A}},
\]
and add it to the approximation above using Fact~\ref{fact}.\ref{fact:sum}.
This turns the $\widehat{\M{A}} \M{D}^{-1} \widehat{\M{A}}$ terms
into $\M{D}$ terms, and completes the proof.
\end{proof}

Now we prove that, with the support  of $\DD-\AA \DD^{-1} \AA$ and $\DD - \tilde{\AA}$, there exists an efficient algorithms to conduct the edge sampling, which leads to the efficient sparsification.

\begin{lemma}\label{lem:matrixCompPlus}
For graph Laplacian $\M{L}_{G}=\DD-\AA$ and $\M{L}_{\tilde{G}_{2r}} = \DD-\tilde{\AA}$ for $r \in \mathbb{Z}_+$, such that (1) $\DD-\tilde{\AA} \approx_{\epsilon'}\DD- \DD\left(\DD^{-1}\AA\right)^{2r}$, and (2) $\tilde{\AA}$ and $\AA$ each contains at most $m$ nonzero entries, then for any $\epsilon > 0$, we can construct
in total work $\OO{m \log^{3} n /\epsilon^4}$, 
a  graph Laplacian $\M{L}_{\tilde{G}_{+}} = \DD-\tilde{\AA}_{+}$  with $\tilde{\AA}_{+}$ containing at most $\OO{n \log n /\epsilon^2 }$ nonzero entries, such that 
\begin{align}
\DD-\tilde{\AA}_{+}
    \approx_{\epsilon'+\epsilon}
  \DD-\DD \left(\DD^{-1}\AA\right)^{2r+4}.
\end{align}
\end{lemma}
\begin{proof}
  By Lemma \ref{lem:symComposition}, we have
  \begin{align}
    \DD - \left( \M{A} \M{D}^{-1} \right)^2 \tilde{\M{A}} \left( \M{D}^{-1} \M{A} \right)^2
    \approx_{\epsilon}
    \DD - \left( \M{A} \M{D}^{-1} \right)^2 \DD \left(  \DD^{-1} \AA \right)^{2r} \left( \M{D}^{-1} \M{A} \right)^2
    = 
    \DD - \DD \left(\DD^{-1}\AA \right)^{2r+4}.
  \end{align}

  Next step is to support $\DD - \left( \M{A} \M{D}^{-1} \right)^2 \tilde{\M{A}} \left( \M{D}^{-1} \M{A} \right)^2$ with 
  $\DD - \tilde{\M{A}}$ 
  and 
  $\DD  - \AA \DD^{-1}\AA$. 

  First, we have 
  \begin{align}
    \exp\left( - \epsilon \right)  \left(\DD  - \AA \DD^{-1}\AA \right)
    \preccurlyeq 
    \exp\left( - \epsilon \right) \left(\DD - \DD \left(\DD^{-1}\AA \right)^{2r+4}\right)
    \preccurlyeq 
    \DD - \left( \M{A} \M{D}^{-1} \right)^2 \tilde{\M{A}} \left( \M{D}^{-1} \M{A} \right)^2
  \end{align}

  We can also show that 
  \begin{align}
  \exp\left( - 2\epsilon \right) \left(\DD - \tilde{\M{A}} \right)
  & \preccurlyeq
    \exp\left( - \epsilon \right) \left(\DD - \DD \left(\DD^{-1}\AA \right)^{2r}\right) \\
  \preccurlyeq
    \exp\left( - \epsilon \right) \left(\DD - \DD \left(\DD^{-1}\AA \right)^{2r+4}\right)
  & \preccurlyeq
    \DD - \left( \M{A} \M{D}^{-1} \right)^2 \tilde{\M{A}} \left( \M{D}^{-1} \M{A} \right)^2.
  \end{align}

 To sample 
$\left( \M{A} \M{D}^{-1} \right)^2 \tilde{\M{A}} \left( \M{D}^{-1} \M{A} \right)^2$
 efficiently, let's consider the length-$5$ path $$\V{p}=\left(  u_0,u_1,u_2,u_3,u_4,u_5 \right),$$ where $\left(u_0,u_1 \right)$, $\left(u_1,u_2 \right)$, $\left(u_3,u_4 \right)$, $\left(u_4,u_5 \right)$ are edges in $\AA$ and $\left(u_2,u_3 \right)$ in $\tilde{\AA}$.

 The edge corresponds to $\V{p}$ has weight $w(\V{p})$ as 
\begin{align}
    w(\V{p}) = \frac{
    \AA\left( u_0,u_1 \right)
    \AA\left( u_1,u_2 \right)
    \tilde{\AA}\left( u_2,u_3 \right)
    \AA\left( u_3,u_4 \right)
    \AA\left( u_4,u_5 \right)   }{ 
    \DD\left( u_1,u_1 \right)
    \DD\left( u_2,u_2 \right) 
    \DD\left( u_3,u_3 \right) 
    \DD\left( u_4,u_4 \right) 
    }.
\end{align} 
And it has upper bound on effective resistance $Z(\V{p})$ as
\begin{align}
    Z(\V{p}) = 
    \frac{\exp(\epsilon)}{\AA\left( u_0,u_1 \right)}
    +
    \frac{\exp(\epsilon)}{\AA\left( u_1,u_2 \right)}
    +
    \frac{\exp(2\epsilon)}{\tilde{\AA}\left( u_2,u_3 \right)}
    +
    \frac{\exp(\epsilon)}{\AA\left( u_3,u_4 \right)}
    +
    \frac{\exp(\epsilon)}{\AA\left( u_4,u_5 \right)}
\end{align} 
The sampling algorithm is similar to that described in Lemma \ref{lem:implicitSample}, with the middle edge random-walk replaced with the random walk on graph $\tilde{\AA}$. Also, the edge index in the path is not uniform among $1,2,\dots,5$ as in Lemma \ref{lem:implicitSample}, i.e., it is now proportional to number of edges in the corresponding random-walk and the additional $\exp(\epsilon)$ terms occurred from the chain of spectral approximation.
However, the total work for path sampling remains the same for $\epsilon = \OO{1}$. After the initial sparsification by path sampling, can we further sparsify the
graph with spectral sparsifiers.
\end{proof}

\subsection{Combining $\M{A}_{+}$ and $\M{A}_{\times}$}

Combining these two components
allows us to prove our main result on sparsifying high degree monomials.


\begin{proof}[Proof of Theorem \ref{thm:power}]
When $d$ is divisible by $4$, we start with $\tilde{\AA}$ such that 
  $\DD-\tilde{\AA}
  \approx_{\epsilon'}
  \DD-\AA\DD^{-1}\AA$,
  and invoke Lemma \ref{lem:matrixCompSquare} and Lemma \ref{lem:matrixCompPlus}
  $k = \OO{\log d}$ times in total, to reach an approximation for
  $\DD-\DD \left(\DD^{-1}\AA\right)^{d}$.
Moreover, if we invoke
  Lemma \ref{lem:matrixCompSquare} and Lemma \ref{lem:matrixCompPlus}
  with approximation parameter $\epsilon/(2k)$,
  and apply \cite{SpielmanTengSpectralSparsification} with $\epsilon/2$ on the final output to obtain $\tilde{G}$,
  the total work is bounded by $\OO{  m \log^{3} n \log^5 (d) /\epsilon^4}$, and $\tilde{G}$ satisfies
  $\M{L}_{\tilde G} \approx_\epsilon \M{L}_{G_d}$.

When $d$ is equal to $4r+2$ for some integer $r$, we first check if $d \le 4/\epsilon$,
  in which case we can directly use Theorem \ref{thm:randWalk} to sparsify $\M{L}_{G_d}$.
When $d > 4/\epsilon$, we can produce a sparsifier for degree $4r$ monomial with error $\epsilon/2$,
  $\M{L}_{\tilde G} \approx_{\epsilon/2} \M{L}_{G_{4r}}$ and use it directly.
This is because
  \begin{align}
    (1 - \lambda^{4r}) \le
    1 - \lambda^{4r+2} \le
    (1 + \frac{1}{2r}) (1 - \lambda^{4r}) \le
    (1 + \frac{\epsilon}{2}) (1 - \lambda^{4r}), \quad
    \forall \lambda \in (-1, 1) \text{ and integer } r.
  \end{align}
So in this situation, we know that $\M{L}_{G_{4r}} \approx_{\epsilon/2} \M{L}_{G_{d}}$.
Combining the two spectral approximation using Fact~\ref{fact}.\ref{fact:transit} gives
  $\M{L}_{\tilde G} \approx_{\epsilon} \M{L}_{G_d}$.
\end{proof}

Note that $G_d$ with even $d$ is usually enough the study the long term effect of the random-walk\footnote{The obvious exception is that when the random-walk is periodic}. However, it is an intriguing question both mathematically and algorithmically if one can sparsify $G_d$ for all $d$ with work polylogarithmic in $d$.


\section{Extension to SDDM Matrices}
\label{sec:generalizeSDDM}

Our path sampling algorithm from Section~\ref{sec:sampling} can
be generalized to a SDDM matrix with splitting $\M{D}-\M{A}$.
The idea is to split out the extra diagonal entries to reduce it back to the Laplacian case.
Of course, the $\M{D}^{-1}$ in the middle of the monomial is changed, however,
  it only decreases the true effective resistance so the upper bound in
  Lemma \ref{lem:upperboundEffRes} still holds without change.
The main difference is that we need to put back the extra diagonal entries,
  which is done by multiplying an all $1$ vector through $\M{D} - \M{D} (\M{D}^{-1} \M{A})^r$.
  
The follow Lemma can be proved similar to Lemma \ref{lem:graphprev}.

\begin{lemma}[SDDM Preservation]
  \label{lem:sddmprev}
If $\M{M} = \M{D} - \M{A}$ is an SDDM matrix with
  diagonal matrix $\M{D}$ and nonnegative off-diagonal $\M{A}$,
  for any nonnegative $\V{\alpha} = (\alpha_1, \ldots, \alpha_d)$
  with $\sum_{r=1}^d \alpha_r = 1$,
  $\M{M}_{\V{\alpha}} = \M{D}-\sum_{r=1}^d \alpha_r \M{D} (\M{D}^{-1} \M{A})^r$
  is also an SDDM matrix.
\end{lemma}

Our algorithm is a direction modification of the algorithm from Section~\ref{sec:sampling}.
To analyze it, we need a variant of Lemma~\ref{lem:upperboundEffRes}
that bounds errors w.r.t. the matrix by which we measure effective resistances. 
We use the following statement from~\cite{Peng13:thesis}.

\begin{lemma}[Lemma B.0.1. from~\cite{Peng13:thesis}]
\label{lem:samplegeneral}
Let $\M{A} = \sum_{i} \V{y}_e^T \V{y}_e$ and $\M{B}$
be $n \times n$ positive semi-definite matrices
such that the image space of $\M{A}$ is contained in the image
space of $\M{B}$, and $\tau$ be a set of estimates such that
\[
\tau_e \geq \V{y}_e^T \M{B}^{\dag} \V{y}_e
\qquad \forall e.
\]
Then for any error $\epsilon$ and any failure probability $\delta = n^{-d}$,
there exists a constant $c_s$ such that if we construct $\M{A}$
using the sampling process from Figure~\ref{fig:ERSparsify},
with probability at least $1 - \delta = 1 - n^{-d}$, $\tilde{\M{A}}$ satisfies:
\[
\M{A} - \epsilon \left( \M{A} + \M{B} \right)
\preceq  \tilde{\M{A}}
\preceq \M{A} + \epsilon \left(\M{A} + \M{B}\right).
\]
\end{lemma}

\begin{theorem}
Let $\M{M}=\M{D}-\M{A}$ be an SDDM matrix with diagonal $\M{D}$,
  nonnegative off-diagonal $\M{A}$ with $m$ nonzero entries,
  for any nonnegative $\V{\alpha} = (\alpha_1, \ldots, \alpha_d)$
  with $\sum_{r=1}^d \alpha_r = 1$, we can define
  $\M{M}_{\V{\alpha}} = \M{D}-\sum_{r=1}^d \alpha_r \M{D} (\M{D}^{-1} \M{A})^r$.
For any approximation parameter $\epsilon>0$,
  we can construct an SDDM matrix $\tilde{\M{M}}$ with $\OO{n \log n / \epsilon^2}$ nonzero entries, in time
  $\OO{m \cdot \log^2 n \cdot d^2 / \epsilon^2}$,
  such that $\tilde{\M{M}} \approx_{\epsilon} \M{M}_{\V{\alpha}}$.
  \end{theorem}
\begin{proof}
We look at each monomial separately.
First, by Lemma \ref{lem:sddmprev}, $\M{M}_r$ is an SDDM matrix.
It can be decomposed as the sum of two matrices,
  a Laplacian matrix $\M{L}_r=\M{D}_r - \M{D} (\M{D}^{-1} \M{A})^r$,
  and the remaining diagonal $\M{D}_{\text{extra}}$.
As in the Laplacian case, a length-$r$ paths in $\M{D} - \M{A}$
  corresponds to an edge in $\M{L}_r$.
We apply Lemma~\ref{lem:samplegeneral} to $\M{M}_r$
  and $\M{L}_r = \sum_{e \in P} \V{y}_e \V{y}_e^\T$,
  where $P$ is the set of all length-$r$ paths in $\M{D} - \M{A}$,
  and $\V{y}_e$ is the column of the incidence matrix associated with $e$.
  
When $r$ is an odd integer, we have
  \begin{align}
  \V{y}_e^\T \left( \M{M}_r \right)^{-1} \V{y}_e \leq
  2 \V{y}_e^\T \left( \M{D}-  \M{A} \right)^{-1} \V{y}_e,
  \end{align}
  and when $r$ is an even integer, we have
  \begin{align}
  \V{y}_e^\T \left( \M{M}_r \right)^{-1} \V{y}_e \leq
  2 \V{y}_e^\T \left( \M{D} - \M{A}\M{D}^{-1} \M{A} \right)^{-1} \V{y}_e.
  \end{align}

Let $e$ denote the edge corresponds to the length-$r$ path $(u_0, \ldots, u_r)$,
  the weight of $e$ is
 \begin{align}
 w(e) = w(u_0 \ldots u_r) = \V{y}_e^\T \V{y}_e =
 \frac {\prod_{i=1}^{r} \M{A} ({u_{i-1}, u_{i}}) } {\prod_{i=1}^{r-1} \M{D} ({u_i, u_i})}
 \le \frac {\prod_{i=1}^{r} \M{A} ({u_{i-1}, u_{i}}) } {\prod_{i=1}^{r-1}  \M{D}_g({u_i, u_i})},
 \end{align}
where $ \M{D}_g  ({u, u}) = \sum_{v \neq u} \M{A}(u, v)$,
  so we have the same upper bound as the Laplacian case,
  and we can sample random walks in the exact same distribution.
%
By Lemma~\ref{lem:samplegeneral} there exists
  $M = \OO{r \cdot m \cdot \log n / \epsilon^2}$ such that
  with probability at least $1-\frac{1}{n}$,
  the sampled graph $\tilde{G}=\algname(G_r, \tau_e, M)$ satisfies
 \begin{align}
 \M{M}_r-\Half \epsilon (\M{L}_r + \M{M}_r) 
 \preccurlyeq 
 \M{L}_{\tilde{G}} & + \M{D}_{\text{extra}}
 \preccurlyeq 
 \M{M}_r+ \Half \epsilon (\M{L}_r + \M{M}_r).
 \end{align}
Now if we set $\M{\tilde{M}} = \M{L}_{\tilde{G}} + \M{D}_{\text{extra}}$, we will have
 \begin{align}
 (1-\epsilon) \M{{M}}_r 
 \preccurlyeq 
\M{\tilde{M}}
 \preccurlyeq 
 (1+\epsilon ) \M{{M}}_r.
 \end{align}
Note that $\M{D}_{\text{extra}}$ can be computed efficiently by
computing $\text{diag}(\M{M}_{r}\V{1})$ via matrix-vector multiplications.
 \end{proof}

\section{Remarks}
\label{sec:remarks}

We gave nearly-linear time algorithms for generating sparse approximations
  of several classes of random-walk matrix-polynomials.
As our study of this problem is motivated by the low
degree case such as for speeding up 
  numerical methods for solving matrix equations, 
  our results only gives part of the picture for the
  high degree case: we are only able to sparsify even-degree
  monomials , and this routine calls
  spectral sparsifiers with error of
  $\epsilon = 1 / \log{d}$ at each step.
Obtaining better algorithms for approximating the structures of
  long-range random-walk matrices is an intriguing mathematical
  and algorithmic question, partially due to the non-commutativity of matrix products.
Extending our algorithm to any $d$, and allowing for higher error
  tolerances at each step are natural directions for future work.
Furthermore, we conjecture that any degree $n$ random-walk
  matrix-polynomial can be sparsified in nearly-linear time.

Our algorithms for the low degree case is based on path sampling.
This routine has analogs in widely used combinatorial network
  analysis routines such as distance estimation~\cite{kangTF09} and
  subgraph counting~\cite{JowhariG05}.
We believe further investigating this connection will lead to
  improved algorithms, as well as models that better explain the
  effectiveness of many existing ones.


\bibliographystyle{alpha}
\bibliography{references}

\newcommand{\etalchar}[1]{$^{#1}$}
\begin{thebibliography}{CKM{\etalchar{+}}11}

\bibitem[AF02]{DavidAldous}
David Aldous and Jim Fill.
\newblock Reversible markov chains and random walks on graphs, 2002.

\bibitem[BSS12]{BSS10}
Joshua Batson, Daniel~A Spielman, and Nikhil Srivastava.
\newblock Twice-ramanujan sparsifiers.
\newblock {\em SIAM Journal on Computing}, 41(6):1704--1721, 2012.

\bibitem[CCL{\etalchar{+}}14]{ChengCLPT14}
Dehua Cheng, Yu~Cheng, Yan Liu, Richard Peng, and Shang{-}Hua Teng.
\newblock Scalable parallel factorizations of {SDD} matrices and efficient
  sampling for gaussian graphical models.
\newblock {\em CoRR}, abs/1410.5392, 2014.

\bibitem[CKM{\etalchar{+}}11]{CKMST10Maxflow}
Paul Christiano, Jonathan~A Kelner, Aleksander Madry, Daniel~A Spielman, and
  Shang-Hua Teng.
\newblock Electrical flows, laplacian systems, and faster approximation of
  maximum flow in undirected graphs.
\newblock In {\em Proceedings of the forty-third annual ACM symposium on Theory
  of computing}, pages 273--282. ACM, 2011.

\bibitem[DS08]{DaitchS08}
Samuel~I. Daitch and Daniel~A. Spielman.
\newblock Faster approximate lossy generalized flow via interior point
  algorithms.
\newblock In {\em Proceedings of the 40th annual ACM symposium on Theory of
  computing}, STOC '08, pages 451--460, New York, NY, USA, 2008. ACM.

\bibitem[GR87]{Greengard}
Leslie Greengard and Vladimir Rokhlin.
\newblock A fast algorithm for particle simulations.
\newblock {\em Journal of computational physics}, 73(2):325--348, 1987.

\bibitem[HL11]{HighamLin}
Nicholas~J Higham and Lijing Lin.
\newblock On $p$th roots of stochastic matrices.
\newblock {\em Linear Algebra and its Applications}, 435(3):448--463, 2011.

\bibitem[JG05]{JowhariG05}
Hossein Jowhari and Mohammad Ghodsi.
\newblock New streaming algorithms for counting triangles in graphs.
\newblock In Lusheng Wang, editor, {\em COCOON}, volume 3595 of {\em Lecture
  Notes in Computer Science}, pages 710--716. Springer, 2005.

\bibitem[KL13]{kelner2013spectral}
Jonathan~A Kelner and Alex Levin.
\newblock Spectral sparsification in the semi-streaming setting.
\newblock {\em Theory of Computing Systems}, 53(2):243--262, 2013.

\bibitem[KMP10]{KMP10SDDSolver}
Ioannis Koutis, Gary~L Miller, and Richard Peng.
\newblock Approaching optimality for solving sdd linear systems.
\newblock In {\em Foundations of Computer Science (FOCS), 2010 51st Annual IEEE
  Symposium on}, pages 235--244. IEEE, 2010.

\bibitem[KOSZ13]{Kelner}
Jonathan~A. Kelner, Lorenzo Orecchia, Aaron Sidford, and Zeyuan~Allen Zhu.
\newblock A simple, combinatorial algorithm for solving sdd systems in
  nearly-linear time.
\newblock In {\em Proceedings of the Forty-fifth Annual ACM Symposium on Theory
  of Computing}, STOC '13, 2013.

\bibitem[KTF09]{kangTF09}
U~Kang, Charalampos~E Tsourakakis, and Christos Faloutsos.
\newblock Pegasus: A peta-scale graph mining system implementation and
  observations.
\newblock In {\em Data Mining, 2009. ICDM'09. Ninth IEEE International
  Conference on}, pages 229--238. IEEE, 2009.

\bibitem[LW12]{LohWainwright}
Po-Ling Loh and Martin~J. Wainwright.
\newblock Structure estimation for discrete graphical models: Generalized
  covariance matrices and their inverses.
\newblock In {\em NIPS}, pages 2096--2104, 2012.

\bibitem[MP13]{MillerP13}
Gary~L Miller and Richard Peng.
\newblock Approximate maximum flow on separable undirected graphs.
\newblock In {\em Proceedings of the Twenty-Fourth Annual ACM-SIAM Symposium on
  Discrete Algorithms}, pages 1151--1170. SIAM, 2013.

\bibitem[Pen13]{Peng13:thesis}
Richard Peng.
\newblock {\em Algorithm Design Using Spectral Graph Theory}.
\newblock PhD thesis, Carnegie Mellon University, Pittsburgh, August 2013.
\newblock CMU CS Tech Report CMU-CS-13-121.

\bibitem[PS14]{PengSpielman}
Richard Peng and Daniel~A. Spielman.
\newblock An efficient parallel solver for sdd linear systems.
\newblock In {\em Proceedings of the 46th Annual ACM Symposium on Theory of
  Computing}, STOC '14, pages 333--342, New York, NY, USA, 2014. ACM.

\bibitem[SF73]{StrangFix}
Gilbert Strang and George~J Fix.
\newblock {\em An analysis of the finite element method}, volume 212.
\newblock Prentice-Hall Englewood Cliffs, NJ, 1973.

\bibitem[SS11]{SpielmanSrivastava}
Daniel~A. Spielman and Nikil Srivastava.
\newblock Graph sparsification by effective resistances.
\newblock {\em SIAM J. Comput.}, 40(6):1913--1926, 2011.

\bibitem[ST04]{SpielmanTengSolver}
Daniel~A. Spielman and Shang-Hua Teng.
\newblock Nearly-linear time algorithms for graph partitioning, graph
  sparsification, and solving linear systems.
\newblock In {\em Proceedings of the Thirty-sixth Annual ACM Symposium on
  Theory of Computing}, STOC '04, pages 81--90, 2004.

\bibitem[ST11]{SpielmanTengSpectralSparsification}
Daniel~A. Spielman and Shang-Hua Teng.
\newblock Spectral sparsification of graphs.
\newblock {\em SIAM J. Comput.}, 40(4):981--1025, July 2011.

\bibitem[ST14]{SpielmanTengLinear}
Daniel~A. Spielman and Shang-Hua Teng.
\newblock Nearly linear time algorithms for preconditioning and solving
  symmetric, diagonally dominant linear systems.
\newblock {\em SIAM J. Matrix Analysis and Applications}, 35(3):835--885, 2014.

\end{thebibliography}

\begin{appendix}
\section{Laplacian Preservation}
\label{apx:graphprev}

\restateGraphPrev*

 \begin{proof}
First note that $\M{L}_{\V{\alpha}}(G) = \sum_{r=1}^d \alpha_r \M{D} (\M{D}^{-1} \M{A})^r$
  is symmetric and has non-positive off-diagonals,
  so to prove that $\M{L}_{\V{\alpha}}(G)$ is also a Laplacian matrix,
  we only need to show the off-diagonals sum to the diagonal.
Fix an integer $r$ and a row index $i$, we study the
  $i$-th row sum $S_{r}$ of $\M{D} (\M{D}^{-1} \M{A})^r$.

 For $r=1$, we have that the row sum $S_{1}$ of $i$-th row of $\M{A}$
   gives $S_{1} = \sum_j {\M{A}}_{i,j} = \M{D}_{i,i}$.
 We show that the row sum $S_{r+1}$
   can be reduce to $S_{r}$ as follows,
 \begin{align}
   S_{r+1} = \sum_{k} \left( \left( \M{D} (\M{D}^{-1} \M{A})^{r} \right)_{i,k} \cdot \M{D}^{-1}_{k,k} \cdot \sum_j \M{A}_{k,j} \right)
   =
    \sum_{k} \left( \M{D} (\M{D}^{-1} \M{A})^{r} \right)_{i,k}
     = S_{r} 
 \end{align}

 By induction, we have that
   $S_{n} = \dots = S_{1} = \Mij{\M{D}}{i}{i}$.
 Thus, the $i$-th row sum of $\M{L}_{\V{\alpha}}(G)$
 \begin{align}
   \sum_j \left( \M{L}_{\V{\alpha}}(G) \right)_{i,j} 
   =
   \sum_{r=1}^t \alpha_r S_r
    =
   \M{D}_{i,i}.
 \end{align}
 Therefore, $\M{L}_{\V{\alpha}}(G)$ is a Laplacian matrix.
 \end{proof}

\section{Support from Linear and Quadratic Terms}
\label{apx:support}

\restateSupportLowOrder*

 \begin{proof}
 Let $\M{X} = \M{D}^{-\frac{1}{2}} \M{A} \M{D}^{-\frac{1}{2}}$,
  for any integer $r$, the statements are equivalent to 
 \begin{align}
   \frac{1}{2} \left(\M{I}-\M{X}\right) 
  \preceq 
   \M{I}-\M{X}^{2r+1} 
  \preceq
    \left(2r+1\right) \left(\M{I}-\M{X}\right) \\
    \M{I}-\M{X}^2
  \preceq
     \M{I}-\M{X}^{2r}
  \preceq
     r \left(\M{I}-\M{X}^2\right).
 \end{align}
 Because $\M{X}$ can be diagonalized by unitary matrix
   $\M{U}$ as $\M{\Lambda} = \M{U} \M{X} \M{U}^{\T}$,
   where $\M{\Lambda} = \text{diag}(\lambda_1,\lambda_2,\dots,\lambda_n)$
   and $\lambda_i \in [-1,1]$ for all $i$.
 Therefore we can reduce the inequalities to the scalar case,
   and we conclude the proof with the following inequalities:
 \begin{align}
 \begin{split}
   \frac{1}{2} (1-\lambda)
 \le 
   1 - \lambda^{2r+1} 
   \le 
   \left(2r+1\right) (1-\lambda) &,
   \quad \forall \lambda \in (-1,1) \text{ and odd integer $r$}; \\
    (1-\lambda^2) 
    \le
     1 - \lambda^{2r} 
    \le
      r (1-\lambda^2) &,
   \quad \forall \lambda \in (-1,1) \text{ and even integer $r$}. \\
 \end{split}
 \end{align}
 \end{proof}

\section{Effective Resistance on Rank One Graph}
\label{apx:supportRankOne}

\begin{proposition}[Claim 6.3. from~\cite{PengSpielman}]
Given a graph of size $n$ with the Laplacian matrix
  $\M{L}=\M{D} - \frac{1}{d} \V{a} \V{a}^{\T}$,
  where $\Mij{\M{D}}{i}{i} = (a_i s) / d $ with $s = \sum_{i=1}^{n}a_i$.
The effective resistance for edge $(i,j)$ is
  \begin{align}
  \frac{d}{s} (\frac{1}{a_i}+\frac{1}{a_j}).
  \end{align}
\end{proposition}
\begin{proof}
Let $\V{e}_i$ denote the vector where the $i$-th entry is 1, and 0 everywhere else. We have
  \begin{align}
  d \M{L} \left(\frac{\V{e}_i}{a_i}-\frac{\V{e}_j}{a_j}\right) = (s-a_i)\V{e}_i-\sum_{k\neq i} a_k \V{e}_k -(s-a_j)\V{e}_j+\sum_{k\neq j} a_k \V{e}_k = s(\V{e}_i -\V{e}_j).
  \end{align}
Therefore
  \begin{align}
  (\V{e}_i -\V{e}_j)^{\T} \M{L}^{\dagger} (\V{e}_i -\V{e}_j) = \frac{d}{s}  (\V{e}_i -\V{e}_j)^{\T}(\frac{\V{e}_i}{a_i}-\frac{\V{e}_j}{a_j}) = \frac{d}{s} (\frac{1}{a_i}+\frac{1}{a_j}). 
  \end{align}
\end{proof}

\end{appendix}

\end{document}